\documentclass[11pt]{article}
\usepackage[letterpaper]{geometry}
\usepackage{fullpage}

\usepackage[dvipdfmx]{color}
\usepackage[dvipdfmx]{graphicx}
\usepackage{tikz}
\usepackage{gnuplot-lua-tikz}

\usepackage[utf8]{inputenc} 
\usepackage[T1]{fontenc}
\usepackage{url}
\usepackage{ifthen}
\usepackage{cite}

\usepackage[cmex10]{amsmath}
\usepackage{amssymb,amsthm}

\newtheorem{theorem}{Theorem}
\newtheorem{corollary}{Corollary}
\newtheorem{lemma}{Lemma}

\newcommand{\ignore}[1]{}

\newcommand{\F}{\mathbb{F}}

\newcommand{\bin}{\{0, 1\}}

\newcommand{\ldis}{{d}_\mathsf{L}}
\newcommand{\lball}{B_\mathsf{L}}

\newcommand{\D}{\mathsf{D}}
\newcommand{\I}{\mathsf{I}}

\newcommand{\out}{\mathrm{out}}
\newcommand{\inin}{\mathrm{in}}
\newcommand{\conc}{\mathrm{conc}}

\begin{document}
%
\title{On the List Decodability of Insertions and Deletions
\thanks{Part of this work appeared in the Proceedings of the 2018 IEEE International Symposium on Information Theory.}}
%
%
%

\author{
  Tomohiro Hayashi\thanks{Kanazawa University, Japan. E-mail: \texttt{t-hayashi@stu.kanazawa-u.ac.jp}}
    \and 
    Kenji Yasunaga\thanks{Osaka University, Japan. E-mail: \texttt{yasunaga@ist.osaka-u.ac.jp}}
}
\maketitle

\begin{abstract}
  In this work, we study the problem of list decoding of insertions and deletions.
  We present a Johnson-type upper bound on the maximum list size.
  The bound is meaningful only when insertions occur.
  Our bound implies that there are binary codes of rate $\Omega(1)$ that are list-decodable from a $0.707$-fraction of insertions.
  For any $\tau_\mathsf{I} \geq 0$ and $\tau_\mathsf{D} \in [0,1)$, there exist $q$-ary codes of rate $\Omega(1)$ that are list-decodable from
    a $\tau_\mathsf{I}$-fraction of insertions and $\tau_\mathsf{D}$-fraction of deletions, where $q$ depends only on $\tau_\mathsf{I}$ and $\tau_\mathsf{D}$.
    We also provide efficient encoding and decoding algorithms for list-decoding
    from $\tau_\mathsf{I}$-fraction of insertions and $\tau_\mathsf{D}$-fraction of deletions
    for any $\tau_\mathsf{I} \geq 0$ and $\tau_\mathsf{D} \in [0,1)$.
    Based on the Johnson-type bound, we derive a Plotkin-type upper bound on the code size in the Levenshtein metric.
  \ignore{
  List decoding of insertions and deletions is studied.
  A Johnson-type upper bound on the maximum list size is derived.
  The bound is meaningful only when insertions occur.
  The result implies that
  there are binary codes  list-decodable from a $0.707$-fraction of insertions.
  For non-binary code, for any constant $\tau_\I \geq 0$ and $\tau_\D \in [0,1)$, there exist codes
    that are list-decodable from a $\tau_\I$-fraction of insertions and a $\tau_\D$-fraction of deletions.
  Efficient encoding and decoding algorithms are also presented for codes with list-decoding radius approaching the above bounds.
  Based on the Johnson-type bound, a Plotkin-type upper bound on the code size in the Levenshtein metric is derived.
  }
\end{abstract}


%

\section{Introduction}

This work addresses the problem of constructing codes for correcting insertions and deletions.
The problem was first studied in the 1960s by Levenshtein~\cite{Lev66}.
Such codes are useful for correcting errors in DNA storage and DNA sequencing~\cite{JHSB17}.

Levenshtein~\cite{Lev66} showed that the code proposed by Varshamov and Tenengolts~\cite{VT65} could be
used to correct a single insertion or deletion.
Multiple insertions and deletions were considered in~\cite{HF02,AGPFC12,PAFC12}.
Codes for correcting a constant fraction of insertions and deletions have been studied
in~\cite{SZ99,GW17,BGH17,GL16}.
Bounds on the size of codes correcting insertions and deletions have been investigated in~\cite{Lev02,KK13,CK14}.

In this work, we consider \emph{list decoding}, in which
the decoder outputs the list of all codewords that lie within some radius of the received word.
Regarding list decoding of insertions and deletions,
Guruswami and Wang~\cite{GW17} construct a binary code that can be list-decoded from
a $(1/2 - \varepsilon)$-fraction of deletions for any $\varepsilon > 0$.
Recently, Wachter-Zeh~\cite{WZ18} claimed a Johnson-type upper bound on the maximum list size
in the presence of insertions and deletions.
However, there is a flaw in an intermediate step in~\cite{WZ18}. We discuss it in Appendix~\ref{sec:wz17}.
We do not know if the flaw can be fixed.

In this paper, we give a corrected Johnson-type bound for insertions and deletions.
Our bound implies that for any $\tau_\I \geq 0$ and $\tau_\D \in [0,1)$, there exist $q$-ary codes list-decodable from
  $\tau_\I$-fraction of insertions and $\tau_\D$-fraction of deletions with a polynomial-sized list, where $q$ depends only on $\tau_\I$ and $\tau_\D$.
  More specifically, when $\tau_\D=0$, 
  polynomial-sized list decoding is possible as long as $\tau_\I < \delta/(1-\delta)$,
  where $\delta$ is the normalized minimum Levenshtein distance of the code.
  The result implies that for any constant $q \geq 2$, the $q$-ary code presented by Bukuh et al.~\cite{BGH17}
  can list decode from $(q+\sqrt{q}-2)/2$-fraction of insertions.
  When $q=2$, the code can correct $0.707$-fraction of insertions.
  Our bound is incomparable with the claimed bound of~\cite{WZ18}. See Section~\ref{sec:comparison} for a detailed comparison.
  We also present efficient encoding and decoding algorithms for list decoding for given radii $\tau_\I \geq 0$ and $\tau_\D \in [0,1)$.
  Finally, based on the Johnson-type bound in the Levenshtein metric,
  we derive a \emph{conditional} Plotkin-type upper bound on the code size.

  


The notion of list decoding was introduced by Elias~\cite{Eli65} and Wozencraft~\cite{Woz58}
in the Hamming metric.
A code is called \emph{$(t, \ell)$-list decodable} if, for any given word $v$,
the number of codewords within Hamming distance $t$ from $v$ is at most $\ell$.
List decoding is meaningful if a large \emph{list-decoding radius} $t$ is attainable
for a small \emph{list size} $\ell$.
Since unique decoding is possible for $t < d/2$, where $d$ is the minimum Hamming distance of the code,
the list-decoding radius $t$ should be at least $d/2$.
For list size $\ell$, we usually require $\ell$ should be a polynomial in the code length.
This is because if $\ell$ is superpolynomial,
it is impossible to guarantee a polynomial-time list decoding.
As an upper bound on the list size in the Hamming metric,
the Johnson bound~\cite{Joh62} guarantees polynomial-size list decoding.
Specifically, it implies that the list size is polynomial in $n$
as long as the list decoding radius is less than $n - \sqrt{n(n-d)}$,
where $n$ is the code length.

Recently, Haeupler, Shahrasbi, and Sudan~\cite{HSS18} presented the construction of list-decodable codes for insertions and deletions
by using an object called \emph{synchronization strings}~\cite{HS17a,HS18,CHLSW19}.
They showed that for every $\tau_\I \geq 0, \tau_\D \in [0,1), \varepsilon > 0$, there exist $q$-ary codes of rate $1-\tau_\D -\varepsilon$
  that are list-decodable from $\tau_\I$-fraction of insertions and $\tau_\D$-fraction of deletions,
  where $q$ depends only on $\tau_\I, \tau_\D, \varepsilon$.
  While their code construction achieves optimal rates, the alphabet size needs to be large constant,
  which is inherent in using synchronization strings.
  In this work, we could obtain binary list-decodable codes.

In Section~\ref{sec:pre}, we introduce the necessary notions and notations.
We prove a Johnson-type bound in Section~\ref{sec:bound}.
In Section~\ref{sec:discussion}, we provide implications from our bound and compare them with the claimed bound
of~\cite{WZ18}.
In Section~\ref{sec:algorithms},
we show efficient encoding and decoding algorithms for codes obtained in the previous sections.
Finally, we present a  Plotkin-type bound in Section~\ref{sec:plotkin}

\section{Preliminaries}\label{sec:pre}

Let $\Sigma$ be a finite alphabet of size $q$.
The Levenshtein distance $\ldis(x, y)$ between two words $x$ and $y$ is
the minimum number of insertions and deletions needed to transform $x$ into $y$.
For a given word $v$, we denote by $\lball(v, t_\I, t_\D)$ the set of words that can be obtained
from $v$ by at most $t_\I$ insertions and at most $t_\D$ deletions.
A code $C \subseteq \Sigma^n$ is said to be \emph{$(t_\I, t_\D, \ell)$-list decodable}
if $|\lball(v, t_\D, t_\I) \cap C | \leq \ell$ for every word $v$.
Note that we consider the set $\lball(v, t_\D, t_\I)$, not $\lball(v, t_\I, t_\D)$.
This is because $x$ is a received word when $t_\I$ insertions and $t_\D$ deletions occurred,
and thus we need to consider codewords that can be obtained from $x$ by $t_\D$ insertions and $t_\I$ deletions.
The minimum Levenshtein distance of a code $C \subseteq \Sigma^n$ is
 the minimum distance $\ldis(c_1, c_2)$ of every pair of distinct codewords $c_1, c_2 \in C$.
Note that the Levenshtein distance between two words in $\Sigma^n$ takes integer values from $0$ to $2n$.
Thus, we normalize the Levenshtein distance by $2n$.
Namely, the \emph{normalized} minimum Levenshtein distance of $C$ is $d/(2n)$, where $d$ is the minimum Levenshtein distance of $C$.
For positive integer $i$, we denote by $[i]$ the set of integers $\{1, \dots, i\}$.

\section{Upper Bound on List Size}\label{sec:bound}

To derive the bound,
we follow the approach of Wachter-Zeh~\cite{WZ18},
which is based on the proofs of Bassalygo~\cite{Bas65} and Roth~\cite[Prop.~4.11]{Roth06}
for the Johnson bound in the Hamming metric.

\begin{lemma}\label{lem:bound}
  Let $C \subseteq \Sigma^n$ be a code of minimum Levenshtein distance $d$.
  For non-negative integers $t_\I$ and $t_\D \leq n$, positive integer $N = n + t_\I - t_\D$, 
  and any word $v \in \Sigma^N$,
  let $\ell = |\lball(v, t_\D, t_\I) \cap C|$ be the list size when $v$ is received.
  If
  \begin{equation}
    \frac{d}{2} > t_\D + \frac{t_\I(n-t_\D)}{N}, \label{eq:listcond}
  \end{equation}
  then 
  \begin{equation}
    \ell \leq \frac{N(d/2 - t_\D)}{N(d/2-t_\D)-t_\I(n-t_\D)}.\label{eq:listupbound}
  \end{equation}
\end{lemma}
\begin{proof}
  Consider the set of codewords $\mathcal{L} \triangleq \lball(v, t_\D, t_\I) \cap C = \{ c_1, c_2, \dots, c_\ell\}$.
  Note that every codeword in $\mathcal{L}$ can be transformed to $v$ by exactly $t_\I$ insertions and $t_\D$ deletions.
  \ignore{
  Namely, the word $v$ can be obtained from each codeword $c_i$
  by at most $t_\I$ insertions and at most $t_\D$ deletions.
  Without loss of generality, we can assume that every codeword in $\mathcal{L}$ can be transformed to $v$
  by exactly $t_\I'$ insertions and $t_\D'$ deletions.
  This is because if $c^* \in \mathcal{L}$ is transformed to $v$ by exactly $t_\I^*$ insertions and $t_\D^*$ deletions with $t_\I^* < t_\I'$ and $t_\D^* < t_\D'$,
  then $v$ can be obtained from $c^*$ by exactly $t_\I'$ insertions and $t_\D'$ deletions,
  in which $t_\I' - t_\I^* ( = t_\D' - t_\D^*)$ symbols are additionally inserted and deleted.
  }
  For each $c_i$, let $\mathcal{D}^{(i)} \subseteq [n]$ and $\mathcal{E}^{(i)} \subseteq [N]$
  be sets of positions for which
  the word obtained from $c_i$ by deleting symbols at positions 
  $\mathcal{D}^{(i)}$ is equal to the word obtained from $v$ by deleting symbols at positions $\mathcal{E}^{(i)}$.
  We can choose $\mathcal{D}^{(i)}$ and $\mathcal{E}^{(i)}$ such that
  $|\mathcal{D}^{(i)}| = t_\D$ and $|\mathcal{E}^{(i)}| = t_\I$ for every $i \in [\ell]$.

  Let $c_i, c_j$ be distinct codewords in $\mathcal{L}$.
  We observe that $c_j$ can be obtained from $c_i$ by
  \begin{enumerate}
    \item deleting symbols at positions $\mathcal{D}^{(i)}$;
  \item inserting some symbols at positions  $\mathcal{E}^{(i)}$ to get $v$;
  \item deleting symbols at positions  $\mathcal{E}^{(j)}$ from $v$;
    and
    \item inserting some symbols at positions $\mathcal{D}^{(j)}$ to get $c_j$.
  \end{enumerate}
  In the above, procedures 2 and 3 can be simplified as follows:
  \begin{enumerate}
  \item[2$'$)] inserting some symbols at positions $\mathcal{E}^{(i)} \setminus \mathcal{E}^{(j)}$;
  \item[3$'$)] deleting symbols at positions $\mathcal{E}^{(j)} \setminus \mathcal{E}^{(i)}$.
    \end{enumerate}
  Thus, we have that
  \begin{equation}
    \ldis(c_i, c_j) \leq |\mathcal{D}^{(i)}| + |\mathcal{E}^{(i)} \setminus \mathcal{E}^{(j)}| + |\mathcal{E}^{(j)} \setminus \mathcal{E}^{(i)}| + |\mathcal{D}^{(j)}|.    \label{eq:dupper}
  \end{equation}

  Next, we consider the sum of distances of all $\ell(\ell-1)$ distinct pair of codewords in $\mathcal{L}$.
  Namely, we define
  $$\lambda \triangleq \sum_{i \in [\ell]} \sum_{j \in [\ell] \setminus \{i\}} \ldis(c_i,c_j).$$
  Since $\ldis(c_i,c_j) \geq d$  for distinct codewords $c_i,c_j$, it holds that
  \begin{equation}
    \lambda \geq \ell(\ell-1)d. \label{eq:lambdalower}
  \end{equation}
  On the other hand, by~(\ref{eq:dupper}), we have that
  \begin{align}
    &\lambda \leq \sum_{i \in [\ell]} \sum_{j \in [\ell] \setminus \{i\}} \left( |\mathcal{D}^{(i)}| + |\mathcal{D}^{(j)}| 
    + |\mathcal{E}^{(i)} \setminus \mathcal{E}^{(j)}| + |\mathcal{E}^{(j)} \setminus \mathcal{E}^{(i)}|
    \right).\label{eq:lambdaupper}
  \end{align}
  Then, it holds that
  \begin{align}
    \sum_{i \in [\ell]} \sum_{j \in [\ell] \setminus \{i\}} \left(|\mathcal{D}^{(i)}| + |\mathcal{D}^{(j)}|\right)  
    & = \sum_{i \in [\ell]} \sum_{j \in [\ell] \setminus \{i\}} \sum_{k \in [n]} \left( I(k \in \mathcal{D}^{(i)}) + I(k \in \mathcal{D}^{(j)}) \right) \nonumber\\
    & = 2 \sum_{k \in [n]} \sum_{i \in [\ell]} \sum_{j \in [\ell] \setminus \{i\}} I(k \in \mathcal{D}^{(i)}) \nonumber \\
    & = 2(\ell-1) \sum_{k \in [n]} X_k, \label{eq:ddupper}
  \end{align}
  where $I(\cdot)$ is the indicator function such that $I(P) = 1$ if predicate $P$ is true, and $I(P) = 0$ otherwise,
  and $X_k = \sum_{i \in [\ell]} I(k \in \mathcal{D}^{(i)})$.
  Similarly,
  \begin{align}
    \sum_{i \in [\ell]} \sum_{j \in [\ell] \setminus \{i\}}\left( |\mathcal{E}^{(i)} \setminus \mathcal{E}^{(j)}| + |\mathcal{E}^{(j)} \setminus \mathcal{E}^{(i)}| \right)
    & = \sum_{k' \in [N]} \sum_{i \in [\ell]} \sum_{j \in [\ell] \setminus \{i\}} \left( I(k'\in \mathcal{E}^{(i)} \setminus \mathcal{E}^{(j)}) 
    + I(k'\in \mathcal{E}^{(j)} \setminus \mathcal{E}^{(i)}) \right)\nonumber\\
    & = 2 \sum_{k' \in [N]} \sum_{i \in [\ell]} \sum_{j \in [\ell] \setminus \{i\}}  I(k'\in \mathcal{E}^{(i)} \setminus \mathcal{E}^{(j)}) \nonumber\\
    & = 2\sum_{k' \in [N]} Y_{k'}( \ell - Y_{k'}), \label{eq:eeupper}
  \end{align}
  where $Y_{k'} = \sum_{i \in [\ell]} I(k' \in \mathcal{E}^{(i)})$.

  It follows from (\ref{eq:lambdalower}), (\ref{eq:lambdaupper}), (\ref{eq:ddupper}), and (\ref{eq:eeupper}) that
  \begin{align}
    \ell(\ell-1)d 
    & \leq 2(\ell-1) \sum_{k \in [n]} X_k + 2\sum_{k' \in [N]} Y_{k'}( \ell - Y_{k'}) \nonumber\\
    & = 2(\ell-1) \sum_{k \in [n]} X_k + 2\ell\sum_{k' \in [N]}Y_{k'} - 2\sum_{k' \in [N]}Y_{k'}^2 \nonumber\\
    & \leq 2(\ell-1) \sum_{k \in [n]} X_k + 2\ell\sum_{k' \in [N]}Y_{k'} - \frac{2}{N}\left( \sum_{k'\in [N]}Y_{k'}\right)^2,
  \end{align}
  where we use the Cauchy-Schwarz inequality in the last step.
  Note that $\sum_{k \in [n]} X_k = \ell t_\D$ and $\sum_{k' \in [N]}Y_{k'} = \ell t_\I$.
  Thus, we have that
  \begin{align*}
    \ell(\ell-1)d  & \leq 2 \ell(\ell-1)t_\D + 2\ell^2 t_\I - \frac{2}{N}\ell^2t_\I^2\\
    & = 2\ell(\ell-1)t_\D + \frac{2\ell^2t_\I(N-t_\I)}{N}.
  \end{align*}
  The inequality can be rewritten as
  \[
  \left( \frac{d}{2} - t_\D - \frac{t_\I(n-t_\D)}{N} \right) \ell^2 \leq \left( \frac{d}{2} - t_\D \right) \ell.
  \]
  Since the coefficient of $\ell^2$ is guaranteed to be positive by~(\ref{eq:listcond}),
  we have that
  \begin{align}
    \ell \leq \frac{N(d/2 - t_\D)}{N(d/2-t_\D)-t_\I(n-t_\D)}.
  \end{align}
  Therefore, the statement follows.
\end{proof}


Lemma~\ref{lem:bound} gives an upper bound on list size when exactly $t_\I$ insertions and $t_\D$ deletions occur.
We evaluate a list-size bound when $t_\I^*$ and $t_\D^*$ are given as upper bounds on the numbers of insertions and deletions, respectively.
It is not difficult to see that (\ref{eq:listcond}) is equivalent to
\begin{equation}
t_\I < \frac{(d/2 - t_\D)(n-t_\D)}{n - d/2}.\label{eq:listcond2}
\end{equation}
Since (\ref{eq:listcond2}) holds only for $t_\D < d/2$, the right-hand side of the inequality is monotonically decreasing on $t_\D$.
Also, the left-hand side is monotonically increasing on $t_\I$.
Thus, if (\ref{eq:listcond2}) is satisfied for $t_\I = t_\I^*$ and $t_\D = t_\D^*$,
then it holds for all non-negative integers $t_\I$ and $t_\D$ satisfying $t_\I \leq t_\I^*$ and $t_\D \leq t_\D^*$.

Regarding the list size bound,
(\ref{eq:listupbound}) is equivalent to
\begin{equation}
  \ell \leq \frac{(n+t_\I-t_\D)(d/2 - t_\D)}{(d/2 - t_\D)(n-t_\D) - (n-d/2)t_\I}.\label{eq:listupbound2}
\end{equation}
If (\ref{eq:listcond2}) is satisfied for $t_\I = t_\I^*$ and $t_\D = t_\D^*$,
the list size is bounded above by
\begin{align*}
 & \max_{t_\I \leq t_\I^*,t_\D \leq t_\D^*} \left\{ \frac{(n+t_\I-t_\D)(d/2 - t_\D)}{(d/2 - t_\D)(n-t_\D) - (n-d/2)t_\I} \right\} 
  \leq \frac{(d/2)(n+t_\I^*)}{(d/2 - t_\D^*)(n-t_\D^*) - (n-d/2)t_\I^*}, 
\end{align*}
where the inequality follows from the fact that the denominator is minimized when $t_\I = t_\I^*$ and $t_\D = t_\D^*$,
and that the numerator is maximized when $t_\I = t_\I^*$ and $t_\D = 0$.

Hence, we have the following theorem.
\begin{theorem}\label{thm:bound}
Let $C \subseteq \Sigma^n$ be a code of minimum Levenshtein distance $d$.
  For non-negative integers $t_\I$ and $t_\D \leq n$, positive integer $N \in [n - t_\D, n + t_\I]$,
  and any word $v \in \Sigma^N$,
  let $\ell = |\lball(v, t_\D, t_\I) \cap C|$ be the list size when $v$ is received.
  If
  \begin{equation}
    t_\I < \frac{(d/2 - t_\D)(n-t_\D)}{n - d/2}\label{eq:listcond3}
  \end{equation}
  then 
  \begin{equation}
    \ell \leq \frac{(d/2)(n+t_\I)}{(d/2-t_\D)(n-t_\D)-(n-d/2)t_\I}.\label{eq:listsize}
  \end{equation}
  
\end{theorem}

\section{Discussion on the Bound}\label{sec:discussion}

First, note that (\ref{eq:listcond3}) is satisfied only when $t_\D < d/2$.
Thus, the bound in Theorem~\ref{thm:bound} is meaningful only when insertions occur because
the list size is at most one when $t_\I = 0$ and $t_\D < d/2$.

By considering the case that $t_\I = t_\D = t$, we have the following corollary.
\begin{corollary}\label{cor:same}
  For a code $C \subseteq \Sigma^n$ of minimum Levenshtein distance $d$,
  a non-negative integer $t < n$, and any word $v \in \Sigma^N$ with $N \in [n-t, n+t]$,
  let $\ell = |\lball(v, t, t) \cap C|$ be the list size when $v$ is received.
  If
  \begin{equation}
    t < n - \sqrt{n(n-d/2)} \triangleq t_\mathsf{equal}
  \end{equation}
  then 
  \begin{equation*}
    \ell \leq \frac{(d/2)(n+t)}{(d/2-2t)n+t^2} \leq nd.
  \end{equation*}
\end{corollary}
\begin{proof}
  By setting $t_\I = t_\D = t$ in (\ref{eq:listcond3}) , we have
  $\left( n - {d}/{2}\right) t < \left( {d}/{2} - t\right) (n-t)$,
  which can be rewritten as
  $(n-t)^2 > n \left( n - {d}/{2} \right)$.
  Since $n - t > 0$, we have that $t < n - \sqrt{n(n-d/2)}$.
  The list size bound is obtained by setting $t_\I = t_\D = t$ in (\ref{eq:listsize}).
\end{proof}


Let $C$ be a code of minimum Levenshtein distance $d$.
Then, the minimum Hamming distance $d'$ of $C$ is at least $d/2$.
If $d' = d/2$, Corollary~\ref{cor:same} gives the original Johnson bound in the Hamming metric.




\begin{figure*}[!t]
    \centering
    \includegraphics[width=0.8\textwidth]{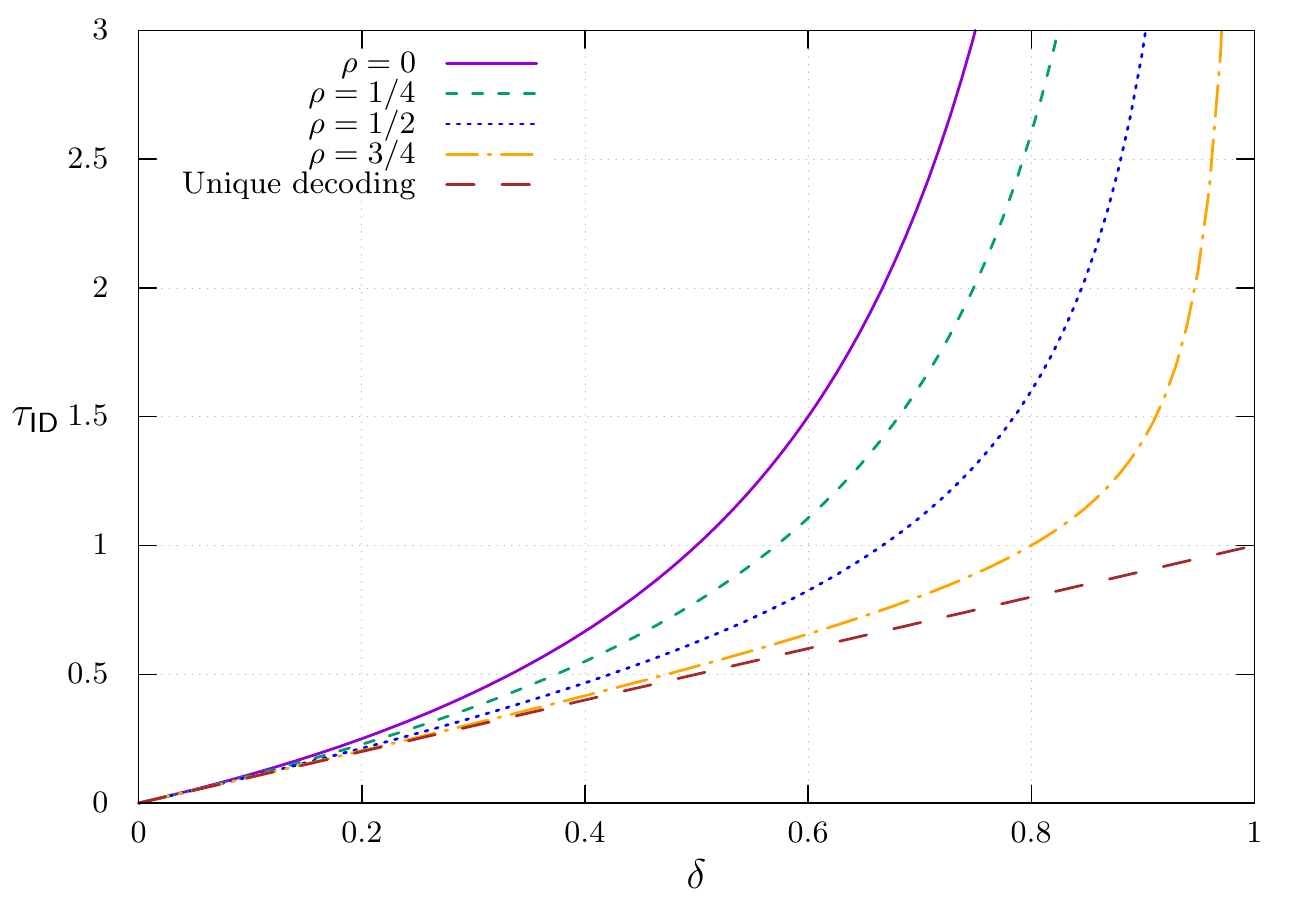}
    \caption{List decoding radius $\tau_\mathsf{ID}$ as a function of $\delta$ for various $\rho$.}
    \label{fig:insdelbound}
  \end{figure*}
  \begin{figure*}[t]
    \centering
    \includegraphics[width=0.8\textwidth]{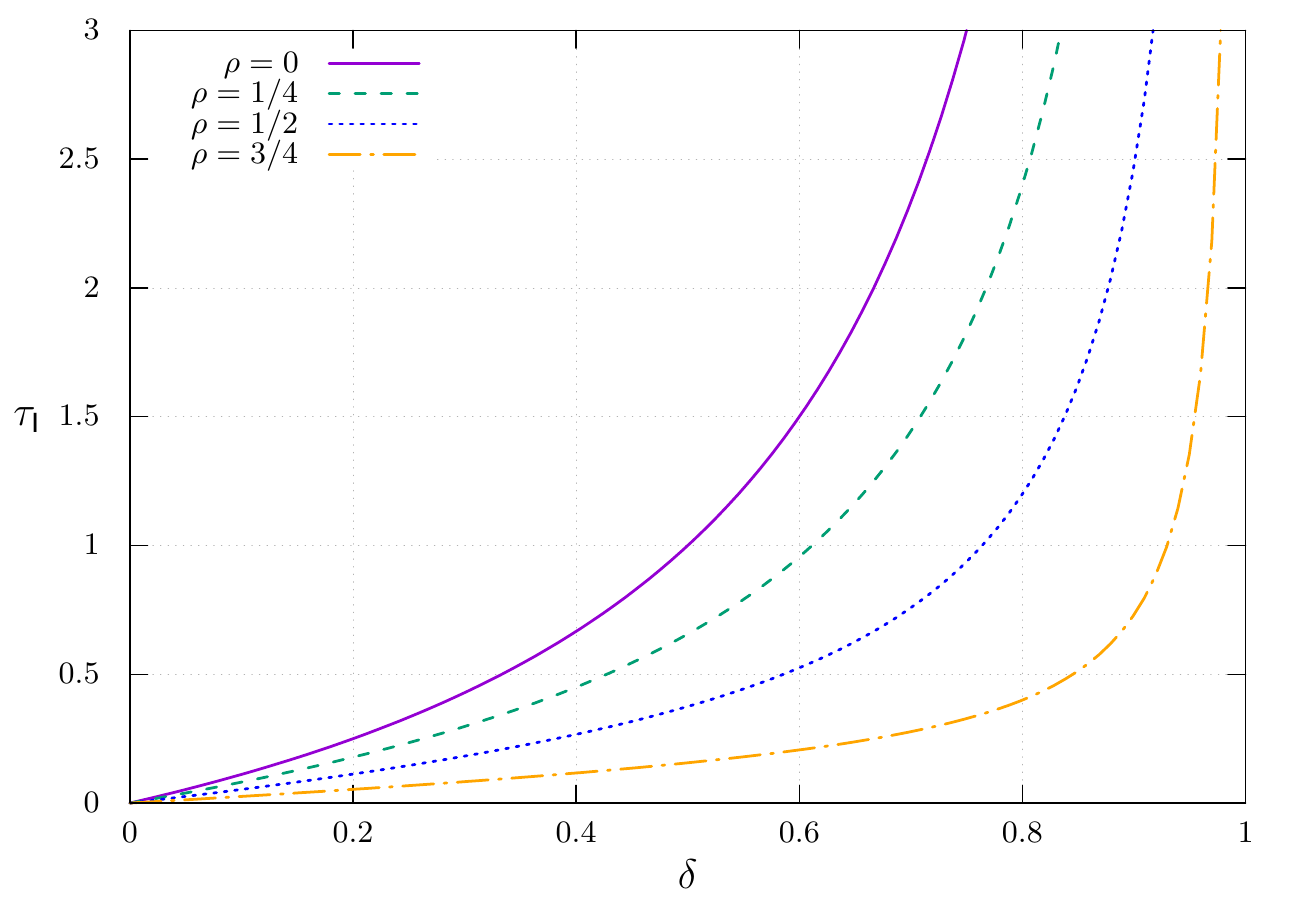}
    \caption{List decoding radius $\tau_\I$ as a function of $\delta$ for various $\rho$.}
    \label{fig:insbound}
  \end{figure*}

  \begin{figure*}[!t]
    \centering
    \includegraphics[width=0.8\textwidth]{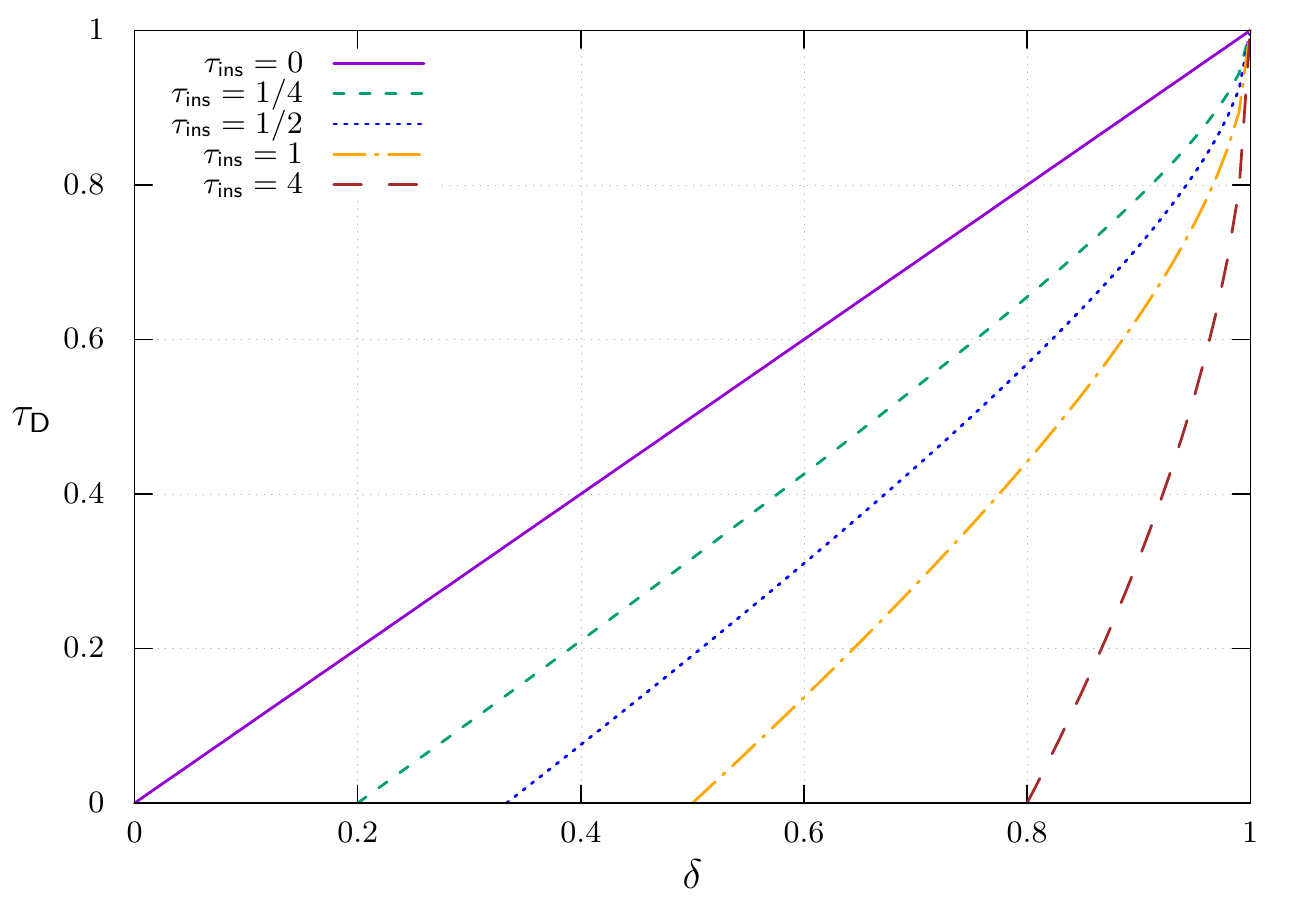}
    \caption{List decoding radius $\tau_\D$ as a function of $\delta$ for various $\tau_\mathsf{ins}$.}
   \label{fig:delbound}
  \end{figure*}

Next, we evaluate the normalized decoding radii $(t_\I + t_\D)/n$, $t_\I/n$, and $t_\D/n$ as functions of
the normalized minimum distance $\delta \triangleq d/(2n)$.
Condition~(\ref{eq:listcond3}) is equivalent to
\begin{align*}
  \frac{t_\I+t_\D}{n} & < \delta + \frac{(1-\rho)^2\delta^2}{1-\delta} \triangleq \tau_\mathsf{ID}(\delta,\rho),\\
  \frac{t_\I}{n} & <  \frac{(1 - \rho) \delta (1 - \rho \delta)}{1 - \delta} \triangleq \tau_\I(\delta,\rho),\\
  \frac{t_\D}{n} & < \frac{1}{2}\left( 1 + \delta - \sqrt{(1-\delta)(1 - \delta + 4 \tau_\mathsf{ins})} \right) \triangleq \tau_\D(\delta, \tau_\mathsf{ins}),
\end{align*}
where $\rho \triangleq t_\D/(d/2) \in [0,1)$ and $\tau_\mathsf{ins} \triangleq t_\I/n \geq 0$.
  The notations used in this section are summarized in Table~\ref{tb:para}.

  \begin{table}[!t]
    \caption{Notations}\label{tb:para}
    \centering
    \begin{tabular}{ll}\hline
      $n$ & Code length \\
      $d \in [0,2n]$ & Minimum Levenshtein distance of the code\\
      $\delta = \frac{d}{2n} \in [0,1]$ & Normalized minimum Levenshtein distance\\
      $t_\I \geq 0$ & Upper bound on the number of insertions\\
      $t_\D \in [0, d/2)$ & Upper bound on the number of deletions\\
        $\tau_\mathsf{ins} \triangleq \frac{t_\I}{n} \geq 0$ & Normalized value of $t_\I$ w.r.t. $n$ \\
        $\rho = \frac{t_\D}{d/2} \in [0,1)$ & Normalized value of $t_\D$ w.r.t. $d/2$ \\
        \hline
    \end{tabular}
  \end{table}

Fig.~\ref{fig:insdelbound} and \ref{fig:insbound} illustrate $\tau_\mathsf{ID}$ and $\tau_\I$ as functions of $\delta$ for various $\rho$.
In Fig.~\ref{fig:delbound}, $\tau_\D$ is represented as a function of $\delta$ for various $\tau_\mathsf{ins}$.

  When only insertions occur, namely $\rho = 0$, we have $\tau_\mathsf{ID}(\delta,0) = \tau_\I(\delta,0) = \delta/(1-\delta)$.
  Thus, the list decoding radius $t_\I/n$ for insertions can take an arbitrarily large value by choosing sufficiently large $\delta$.

  For any fixed constant $q$, Bukh et al.~\cite{BGH17} construct a $q$-ary code of rate $\Omega(1)$ that can correct a fraction of deletions approaching
  $1- \frac{2}{q+\sqrt{q}}$,
which implies that the normalized minimum Levenshtein distance $\delta$ of the code is also approaching $1- \frac{2}{q+\sqrt{q}}$.
For this code, when $\rho = 0$, the list decoding radius is $\tau_\mathsf{I}(\delta,0) = \frac{q + \sqrt{q}}{2} - 1$.
This implies that when $q=2$, the code is potentially list-decodable from a $0.707$-fraction of insertions in polynomial time.

We can see that the radius $t_\D/n$ for deletions can also take any value less than $1$.
\begin{corollary}\label{cor:summary}
    For any  $\tau_\I \geq 0$ and $\tau_\D \in [0,1)$,
      if there is a code of length $n$ and  minimum Levenshtein distance $d = 2 \delta n$ satisfying
      \begin{equation}
        \delta > \frac{\tau_\I+\tau_\D(1-\tau_\D)}{\tau_\I+1-\tau_\D} = 1 - \frac{(1-\tau_\D)^2}{\tau_\I+1-\tau_\D} \triangleq \delta_\mathsf{ID}(\tau_\I,\tau_\D),\label{eq:distlb}
      \end{equation}
      then the code is $(\tau_\I n, \tau_\D n, \ell)$-list decodable for 
      \[ \ell \leq \frac{\delta(\tau_\I+1)}{\gamma \, (\tau_\I + 1-\tau_\D)},\]
      where $\gamma = \delta - \delta_\mathsf{ID}(\tau_\I,\tau_\D) > 0$.
\end{corollary}
\begin{proof}
  Let $\rho = \tau_\D/\delta$.
  Then, the inequality~(\ref{eq:distlb}) is equivalent to $\tau_\I < \tau_\I(\delta, \rho)$.
  Thus, by the discussion on the above, the code is $(\tau_\I n, \tau_\D n, \ell)$-list decodable for some $\ell$.
  It follows from~(\ref{eq:listsize}) that $\ell$ is bounded above by
  \begin{align*}
    \frac{\delta(\tau_\I+1)}{(\delta - \tau_\D)(1-\tau_\D) - (1 - \delta)\tau_\I}
    & = \frac{\delta(\tau_\I+1)}{\gamma \,(\tau_\I + 1-\tau_\D)}.
  \end{align*}
\end{proof}

It is known that, for any $\varepsilon > 0$, the normalized minimum distance $1 - \varepsilon$ is attainable by codes over an alphabet of size $\Omega( \varepsilon^3)$.
\begin{lemma}[Corollary~11 of~\cite{GW17} for $\theta = 1/3$]\label{lem:highdistcode}
  Let $\varepsilon \in (0,1/2)$. For every $n$, there exists a code $C \subseteq [q]^n$ of rate $R = \varepsilon/3$
  such that the minimum Levenshtein distance of $C$ is at least $2(1 - \varepsilon)n$, provided $q \geq 64/\varepsilon^3$.
  Moreover, $C$ can be constructed, encoded, and decoded in time  $q^{O(n)}$.
\end{lemma}

Hence, given $\tau_\I \geq 0$ and $\tau_\D \in [0,1)$,
  we can obtain a $(\tau_\I n, \tau_\D n, \ell)$-list decodable code
  by choosing  $\delta = 1 - \varepsilon + \gamma$ for $\varepsilon = (1-\tau_\D)^2/(\tau_\I+1-\tau_\D)$ and small $\gamma > 0$.


\subsection{Comparison with the claimed bounds of~\cite{WZ18}}\label{sec:comparison}

For a code $C \subseteq \Sigma^n$ of minimum Levenshtein distance $d$,
it is claimed in~\cite[Corollary~1]{WZ18} that
for any integer $t$ and word $v \in \Sigma^N$, if $t < n+N - \sqrt{(n+N)(n+N-d)} \triangleq t_\mathsf{WZ}$,
then $\ell = |\lball(v,t,t) \cap C| \leq (n+N)d/(t^2-(2t-d)(n+N))$.
Compared to Theorem~\ref{thm:bound}, this bound is restricted to the case that the upper bounds on the numbers of insertions and deletions are the same.
The case that $t_\I = t_\D$ in Theorem~\ref{thm:bound} yields Corollary~\ref{cor:same}.
The normalized bound $t_\mathsf{WZ}/n =  1+N/n - \sqrt{(1+N/n)(1+N/n-2\delta)}$ of~\cite{WZ18} is
better than $t_\mathsf{equal}/n =  1 -\sqrt{1-\delta}$ of Corporeally~\ref{cor:same}, where $\delta = d/2n$.
Note that Theorem~\ref{thm:bound} gives a Johnson-type bound without the restriction that $t_\I = t_\D$.

In~\cite{WZ18}, Johnson-type bounds were also claimed when only deletions or only insertions occur.
As observed above, when only deletions occur, Theorem~\ref{thm:bound} yields a trivial bound of unique decoding.
It remains open to giving a Johnson-type bound when only deletions occur.
For the case that only insertions occur, while the claimed bound of~\cite[Corollary~3]{WZ18} has an upper bound on the fraction of insertions $\tau_\I$,
our bound of Corollary~\ref{cor:summary} can take any value $\tau_\I \geq 0$.


\section{Efficient Encoding and Decoding Algorithms}\label{sec:algorithms}

In this section, we present efficient encoding and decoding algorithms  of list-decodable codes for any given radii $\tau_\I = t_\I/n \geq 0$ and $\tau_\D =t_\D/n \in [0,1)$.
  The construction is based on concatenated codes with outer Reed-Solomon code, which are also used in  work~\cite{GW17,GL16} for
  explicit codes for \emph{unique} decoding of insertions and deletions.
  Our observation is that a list-decoding algorithm of insertions and deletions can be constructed in a similar way as unique decoding of~\cite{GW17,GL16}.
  
The outer code is a Reed-Solomon code of length $n$ and rate $r$ over $\F_p$ with $p \geq n$.
Specifically, the code $C_\out : \F^{r n}_p \to \F^{n}_p$ is defined as the function that maps
message $s = (s_1, \dots, s_{r n}) \in \F_p^{r n}$ to the codeword $(f_s(\alpha_1), \dots, f_s(\alpha_{n}))$,
where $f_s(X) \triangleq s_1 + s_2 X + \dots + s_{r n} X^{r n-1}$ and $\alpha_1, \dots, \alpha_{n}$ are distinct elements of $\F_p$.
The inner code is 
a code $C_\inin : \F_p \times \F_p \to [q]^m$ of rate $r'$. 
Since the rate of $C_\inin$ is $r'$, we have the relation $m = (\log_q p^2)/r' = 2\log_2 p/(r' \log_2 q)$. 
The resulting code $C_\conc: \F_p^{r n} \to [q]^{nm}$ is defined such that for message $s \in \F_p^{r n}$,
the final codeword is $C_\conc(s) \triangleq (C_\inin(\alpha_1, f_s(\alpha_1)), \dots, C_\inin(\alpha_{n}, f_s(\alpha_{n})))$.
The rate of $C_\conc$ is $(r r')/2$. 

We assume that $C_\inin$  is $(\tau_\I' m, \tau_\D' m, \ell')$-list decodable for
$\ell' = O(1)$, and the encoding can be done in time $T(m)$.
It follows from Corollary~\ref{cor:summary} that 
the constant list size can be achieved  as long as
the normalized minimum Levenshtein distance $\delta$ of the code satisfies $\delta \geq \delta_\mathsf{ID}(\tau_\I',\tau_\D') + \Omega(1)$.
Note that the encoding of the concatenated code can be done in time polynomial in $n$ when $T(m) = q^{O(m)}$ for $q = O(1)$ and $m = O(\log_2 n)$.


We show that 
$C_\conc$ is list-decodable for insertions and deletions in polynomial time.


\begin{lemma}\label{lem:explicit}
  For any $\tau_\I \geq 0$ and $\tau_\D \in [0,1)$,
    the above code $C_\conc$ is $(\tau_\I nm, \tau_\D nm, \ell)$-list decodable
    by choosing $\tau_\D' = \sqrt{\tau_\D}$, $\tau_\I' = 2\tau_\I(1-\tau_\D')^{-1}+(1-\tau_\D')/2$, $r = {(1-\tau_\D')^4}/{(32(1+\tau_\I')^2\ell')}$,
    and $\ell = {16 (1+\tau_\I')^2 }{(1-\tau_\D')^{-3}}\cdot \ell'$.
  The decoding can be done in time $O\left( n^3 m^2 \cdot T(m)\right)$. 
\end{lemma}
\begin{proof}
  First, we specify the decoding algorithm,
  where the input is a string $v$ that is obtained by changing a codeword $c$ under
  at most $\tau_\I nm$ insertions and at most $\tau_\D nm$ deletions.
  The idea is to list-decode the inner codewords by testing sufficiently many ``windows'' and apply the list-recover algorithm of the outer Reed-Solomon code.
  On input $v \in [q]^{*}$ of length from $(1-\tau_\D)nm$ to $(1+\tau_\I)nm$,
  \begin{enumerate}
  \item Set $\mathcal{J} \gets \emptyset$.
  \item For each $0 \leq j \leq \left\lceil \frac{2(1+\tau_\I)n}{1-\tau_\D'} \right\rceil$ and
    $1 \leq j' \leq \left\lceil \frac{2(1+\tau_\I')}{1-\tau_\D'} \right\rceil$,
    do the following:\label{step:2}
    \begin{enumerate}
    \item Let $v[j,j']$ denote the substring of $v$ indices from
      $\left\lfloor \frac{(1-\tau_\D')m}{2}\right\rfloor j+1$ to $\left\lfloor \frac{(1-\tau_\D') m}{2}\right\rfloor (j+j')$; 
    \item For every pair $(\alpha, \beta) \in \F_p \times \F_p$,
      if $v[j,j']$ can be obtained from $C_\inin(\alpha,\beta)$ by at most $\tau_\I' m$ insertions and at most $\tau_\D' m$ deletions,
      add $(\alpha, \beta)$ to $\mathcal{J}$.\label{step:2-b}
      \end{enumerate}
  \item Find the list $\mathcal{L}$ that contains all polynomials $f$ of degree less than $rn$ for which
    $\left|\{ (\alpha, f(\alpha)) : \alpha \in \F_p\} \cap \mathcal{J}\right| \geq \frac{1-{\tau_\D'}}{2}n$.\label{step:3}
  \item Output the list of messages $s \in \F_p^{rn}$ such that
    $f_s \in \mathcal{L}$ and $v$ can be obtained from $C_\conc(s)$ by at most $\tau_\I nm$ insertions and at most $\tau_\D nm$ deletions.
  \end{enumerate}

  Next, we discuss the correctness of the decoding.
  For the codeword $c = C_\conc(s)$ for some $s \in \F_p^{r n}$, parse $c$ into $n$ blocks $(c_1, \dots, c_{n})$,
  where $c_i = C_\inin(\alpha_i, f_s(\alpha_i))$.
  Consider that the received word $v$ is parsed into $n$ blocks
  such that the $i$-th block $v_i$ is obtained from $c_i$.
  We say index $i \in [n]$ is \emph{good} if $v_i$ is obtained from $c_i$
  by at most $2\tau_\I(1-\tau_\D')^{-1} m$ insertions and at most $\tau_\D' m$ deletions, and $\emph{bad}$ otherwise.
  Since the number of deletions occurred in generating $v$ from $c$ is at most $\tau_\D nm$,
  the number of indices $i$ for which more than $\tau_\D' m$ deletions occurred in generating $v_i$ from $c_i$ is at most
  \begin{equation*}
    \frac{\tau_\D nm}{\tau_\D'm}  = {\tau_\D'} n.
  \end{equation*}
  Hence, $(1-{\tau_\D'})$-fraction of $v_i$'s are obtained from $c_i$ by at most $\tau_\D'm$ deletions.
  Similarly, the number of $v_i$'s for which more than $2\tau_\I(1-\tau_\D')^{-1} m$ insertions occurred is at most
  \[ \frac{\tau_\I nm}{2\tau_\I(1-\tau_\D')^{-1} m} = \frac{1-{\tau_\D'}}{2}n.\]
  Thus, the number of good indices are at least $(1-\tau_\D')n - \frac{1-{\tau_\D'}}{2}n = \frac{1-{\tau_\D'}}{2}n$. 

  For any good index $i$, there exists $v[j,j']$ such that $0 < |v[j,j']| - |v_i| < \frac{(1-\tau_\D')m}{2}$ and
  $v_i$ is a substring of $v[j,j']$.
  Then, $v[j,j']$ can be obtained from $c_i$ by at most $2\tau_\I(1-\tau_\D')^{-1} m + \frac{(1-\tau_\D')m}{2} = \tau_\I' m$ insertions and
  at most $\tau_\D' m$ deletions.
  Since $c_i = C_\inin(\alpha_i, f_s(\alpha_i))$, the pair $(\alpha_i, f_s(\alpha_i))$ is added to $\mathcal{J}$
  in Step~\ref{step:2}.
  Namely, for correct message $s$ and its good index $i$, the pair $(\alpha_i, f_s(\alpha_i))$ will be included in $\mathcal{J}$.
  Since the number of good indices is at least $\frac{1-{\tau_\D'}}{2}n$, the polynomial $f_s$ will appear in $\mathcal{L}$ by Step~\ref{step:3}.
  
  For Step~\ref{step:3}, we employ Sudan's list decoding algorithm~\cite{Sud97}.
  The algorithm, on input a set $\mathcal{J} \subseteq \F_p \times \F_p$,
  outputs a list $\mathcal{L}$ of all polynomials $f$ of degree less than $r n$ such that
  $(\alpha, f(\alpha)) \in \mathcal{J}$ for more than $\sqrt{2(r n)|\mathcal{J}|}$ values of $\alpha \in \F_p$.
  It runs in time polynomial in  $|\mathcal{J}|$ and $\log p$, and the size of $\mathcal{L}$ is at most
  $\sqrt{2|\mathcal{J}|/(r n)}$.
  For each $v[j,j']$, the number of codewords in $C_\inin$ that can be obtained from $v[j,j']$ by at most $\tau_\I' m$ insertions and $\tau_\D' m$ deletions is
  at most $\ell'$.
  Hence, the size of $\mathcal{J}$ is bounded by
  \begin{align*}
    | \mathcal{J} | & \leq \left\lceil  \frac{2(1+\tau_\I)n}{1-\tau_\D'} \right\rceil \cdot \left\lceil \frac{2(1+\tau_\I')}{1-\tau_\D'}\right\rceil
    \cdot \ell' \leq \frac{4n\ell'(1+\tau_\I')^2}{(1-\tau_\D')^2}.
  \end{align*}
  The agreement parameter in Step~\ref{step:3} satisfies
  \begin{align*}
    \frac{1-{\tau_\D'}}{2}n =
    \sqrt{ 2 \cdot \frac{(1-\tau_\D')^4}{32(1+\tau_\I')^2\ell'} n \cdot \frac{4n\ell'(1+\tau_\I')^2}{(1-\tau_\D')^2}}
    \geq \sqrt{2(r n)|\mathcal{J}|}. 
  \end{align*}
  Thus, the algorithm can correctly output the list $\mathcal{L}$ in Step~\ref{step:3}.
  The size of $\mathcal{L}$ is at most $\sqrt{2|\mathcal{J}|/(rn)} \leq {16 (1+\tau_\I')^2 }{(1-\tau_\D')^{-3}}\cdot \ell' = \ell$.


  Finally, we evaluate the running time of decoding.
  In Step~\ref{step:2}, each $v[j,j']$ is of length at most $(1+\tau_\I') m$, and we check $O(n)$ combinations of $(j,j')$.
  The encoding time of $C_\inin(\alpha,\beta)$ is $T(m)$,
  and the Levenshtein distance can be computed in time $O( m^2)$.
  Since we search for all $(\alpha,\beta) \in \F_p \times \F_p$ for each $v[j,j']$,
  Step~\ref{step:2} takes time $O(n) \cdot p^2 \cdot T(m) \cdot O(m^2) = O( p^3m^2) \cdot T(m)$.
  It is known that when $|\mathcal{J}| = O(n)$ and $\ell = O(1)$,
  Sudan's list decoding can be performed in time $O(n^2)$~\cite{RR00}, and hence Step~\ref{step:3} takes time $O(n^2)$.
  Since the list size is at most $\ell = O(1)$, the final step can be performed in time $O((nm)^2)$.
  By setting $p = n$, the overall running time of decoding is $O(n^3m^2\cdot T(m)) + O(n^2)+ O((nm)^2) = O(n^3 m^2 \cdot T(m) )$. 
\end{proof}

We have the following theorem.
\begin{theorem}\label{thm:constalphabetcode}
  For any constants $\tau_\I \geq 0$ and $\tau_\D \in [0,1)$,
    there exists a code $C \subseteq [q]^N$ of rate $\Omega(1)$ for $q = O(1)$
    that is $(\tau_\I N, \tau_\D N, O(1))$-list decodable.  The code can be decoded in time $N^{O(1)}$.
\end{theorem}
\begin{proof}
  As discussed in Section~\ref{sec:discussion}, for any $\tau_\I$ and $\tau_\D$,
  a $(\tau_\I'm, \tau_\D'm,O(1))$-list decodable code $C_\inin$ described above
  can be constructed by choosing sufficiently small constant $\varepsilon$ in Lemma~\ref{lem:highdistcode}.
  Then, the rate of $C_\conc$ is $(rr')/2 = \Omega(1)$, and $q = O(1)$.
  Since $C_\inin$ can be encoded in time $q^{O(m)}$ by Lemma~\ref{lem:highdistcode},
  the decoding of $C_\conc$ can be done in time $O(n^3 m^2) \cdot  q^{O(m)} = O(n^3 m^2) \cdot n^{O(1/\varepsilon)} = N^{O(1)}$. 
\end{proof}

Although the constant hidden in $O(1)$ of the decoding complexity may be large for small $\varepsilon$,
the time complexity can be reduced by employing the code of Theorem~\ref{thm:constalphabetcode} as the inner code.

\subsection*{Insertion-Only Case}

If only insertions occur, the resulting concatenated code can achieve almost the same list-decoding radius as that of the inner code.
This observation is crucial for constructing list-decodable codes over small alphabets.

The construction is the same as the previous one.
We assume that the inner code $C_\inin$ over $[q]$ of length $m$ is $(\tau_\I', 0, \ell')$-list decodable
for $q = O(1)$ and $\ell'=O(1)$,
and can be encoded in time $T(m)$.
\begin{lemma}\label{lem:insertcode}
  For any $\tau_\I \geq 0$, $\gamma \in (0,1)$, and integer $k \geq 2$,
  the code $C_\conc$ is $(\tau_\I nm, 0, \ell)$-list decodable by choosing
  $\tau_\I' = (1+\gamma)\left( \tau_\I + \left\lceil (1+\tau_\I)/k \right\rceil \right)$, $r= \gamma^2/(8k^2\ell')$, and
  $\ell = 4k^2\ell'/\gamma$.
  The decoding can be done in time $O(n^3m^2 \cdot T(m))$.
\end{lemma}
\begin{proof}
  The decoding algorithm is changed as follows.
  On input $v \in [q]^*$ of length from $nm$ to $(1+\tau_\I)nm$,
  \begin{enumerate}
  \item Set $\mathcal{J} \gets \emptyset$, and let $b = \left\lceil \frac{(1+\tau_\I)m}{k} \right\rceil$.
  \item \label{step:binary2}For each $0 \leq j \leq \left\lfloor \frac{(1+\tau_\I)nm}{b} \right\rfloor$ and
    $1 \leq j' \leq k$,
    do the following:
    \begin{enumerate}
    \item Let $v[j,j']$ denote the substring of $v$ indices from
      $b  j+1$ to $b (j+j')$; 
    \item For every pair $(\alpha, \beta) \in \F_p \times \F_p$,
      if $v[j,j']$ can be obtained from $C_\inin(\alpha,\beta)$ by at most $\tau_\I' m$ insertions,
      add $(\alpha, \beta)$ to $\mathcal{J}$.
      \end{enumerate}
  \item \label{step:binary3}Find the list $\mathcal{L}$ that contains all polynomials $f$ of degree less than $rn$ for which
    $\left|\{ (\alpha, f(\alpha)) : \alpha \in \F_p\} \cap \mathcal{J}\right| \geq \frac{\gamma}{2}n$.
  \item Output the list of messages $s \in \F_p^{rn}$ such that
    $f_s \in \mathcal{L}$ and $v$ can be obtained from $C_\conc(s)$ by at most $\tau_\I nm$ insertions.
  \end{enumerate}

  As in the proof of Lemma~\ref{lem:explicit},
  we assume that received word $v$ can be parsed into $(v_1, \dots, v_n)$ such that
  $v_i$ is obtained by the $i$-th block of the original codeword $c  = (c_1, \dots, c_n)$, where $c_i = C_\inin(\alpha_i,f_s(\alpha_i))$ for some $s \in \F_p^{rn}$.
  We say index $i \in [n]$ is \emph{good} if $v_i$ is obtained from $c_i$ by at most $\tau_\I' m$ insertions, and \emph{bad} otherwise.
  Then, the number of indices $i$ for which more than $\tau_\I' m$ insertions occurred is at most
  \begin{equation*}
    \frac{\tau_\I nm}{\tau_\I'm}  \leq \frac{n}{(1+\gamma)(1+1/k)} < \frac{n}{1+\gamma} = \left(1 - \frac{1}{\gamma^{-1} + 1}\right)n< \left(1-\frac{\gamma}{2}\right)n
  \end{equation*}
  for $\gamma \in (0,1)$.
  Thus, a $\gamma/2$-fraction of $v_i$'s are obtained from $c_i$ by at most $\tau_\I' m$ insertions, and their indices are good.

  For any good index $i$, there exists $v[j,j']$ such that $0 < \left| v[j,j']\right| - |v_i|  < b$ and
  $v_i$ is a substring of $v[j,j']$.
  Then, $v[j,j']$ can be obtained from $c_i$ by at most $\tau_\I m + b \leq \tau_\I' m$ insertions.
  Hence, for correct message $s$ and its good index $i$, the pair $(\alpha_i, f_s(\alpha_i))$ will be included in $\mathcal{J}$.
  Since the number of good indices is at least $\frac{\gamma}{2}n$, the polynomial $f_s$ will appear in $\mathcal{L}$.

  We employ Sudan's list decoding in Step~\ref{step:binary3}.
    It follows from a similar argument as the proof of Lemma~\ref{lem:explicit} that
    the size of $\mathcal{J}$ is bounded by
  \begin{align*}
    | \mathcal{J} | & \leq \left\lfloor \frac{(1+\tau_\I)nm}{b} \right\rfloor \cdot k \cdot \ell' \leq n k^2\ell'
  \end{align*}
  and the agreement parameter in Step~\ref{step:binary3} satisfies
  \begin{align*}
    \frac{\gamma}{2} n = \sqrt{2 \cdot \frac{\gamma^2}{8k^2\ell'} n \cdot nk^2 \ell'}
    \geq \sqrt{2(r n)|\mathcal{J}|}. 
  \end{align*}
  Thus, the algorithm can correctly output the list $\mathcal{L}$ in Step~\ref{step:3}.
  The list size is at most $\sqrt{2|\mathcal{J}|/(rn)} \leq 4k^2\ell'/\gamma = \ell$.

  We evaluate the decoding complexity.
  Note that each $v[j,j']$ in Step~\ref{step:binary2} is of length at most $bk = O(m)$,
  and that $|\mathcal{J}| = O(n)$ and $\ell = O(1)$.
  It follows from the same argument as in the proof of Lemma~\ref{lem:explicit} that, by setting $p = n$,
  the decoding complexity is $O(n^3m^2 \cdot T(m))$.
\end{proof}

By employing the code of~\cite{BGH17} as the inner code of the concatenated code in Lemma~\ref{lem:insertcode},
we have the following theorem.

\begin{theorem}\label{thm:binary}
  For any integer $q \geq 2$ and $\gamma_1 \in (0,1]$, there are infinitely many and sufficiently large $N$
for which there exists a code $C \subseteq [q]^N$ of rate $\Omega(1)$ 
  that is $( \tau_\I N, 0, O(1))$-list decodable for $\tau_\I = (q+\sqrt{q})/2 - 1 - \gamma_1$.
  The code can be decoded in time $O(N^3(\log_2\log_2 N)^2)$.
\end{theorem}
\begin{proof}
  We use the following code as the inner code $C_\inin$.
  \begin{lemma}[Theorem~19 of~\cite{BGH17}]
    Fix an integer $q \geq 2$ and real $\gamma_2 > 0$.
    There are infinitely many and sufficiently large $m$ for which
    there is a code $C_\inin \subseteq [q]^m$ of rate $(\gamma_2/q)^{O(\gamma^{-3})}$ such that
    the normalized minimum Levenshtein distance is at least $1 - 2/(q+\sqrt{q}) - \gamma_2$,
    and the encoding can be done in time $O(m \log_2^2 m)$.
  \end{lemma}
  Let $\delta = 1 - 2/(q+\sqrt{q}) - \gamma_2$ be the normalized distance of the code in the above lemma.
  It follows from the discussion in Section~\ref{sec:discussion} that the code is list-decodable against $\delta/(1-\delta)$-fraction of insertions.
  Then, the list-decoding radius is
  \[ \frac{\delta}{1-\delta} = \frac{q+\sqrt{q}}{2} - 1 - \frac{(q+\sqrt{q})^2}{4/\gamma_2 + 2(q+\sqrt{q})},\]
  which can be arbitrarily close to $(q+\sqrt{q})/2 - 1$ for sufficiently small $\gamma_2$.
  For any given $\gamma_1 > 0$, by choosing sufficiently small $\gamma$ and large $k$ in Lemma~\ref{lem:insertcode},
  we can construct a $(\tau_\I N, 0, \ell)$-list decodable code for
  $\tau_\I = (q+\sqrt{q})/2 - 1 - \gamma_1$ and $\ell = 4k^2\ell'/\gamma = O(1)$, where $N = nm$.

  The rate of the code $C_\conc$ is $\Omega(1)$ since both $r = \gamma^2/(8k^2\ell')$ and $r' =(\gamma_2/q)^{O(\gamma^{-3})}$ are constant.
  Since $m = O(\log_2 n)$, the list-decoding of $C_\conc$ can be done in time $O(n^3m^2 \cdot m\log_2^2m) = O(N^3(\log_2\log_2 N)^2)$.
 \end{proof}


By Theorem~\ref{thm:binary}, we obtain a binary code of rate $\Omega(1)$ that is list-decodable from $(2+\sqrt{2})/2-1 \approx 0.707$-fraction of insertions
in polynomial time. 

\section{Upper Bound on Code Size}\label{sec:plotkin}

In this section, we present a Plotkin-type bound on code size in the Levenshtein metric.
The Plotkin bound~\cite{Plo60} claims that if the relative minimum distance of a given code is large,
the code must be small.
More specifically, for any code $C \subseteq [q]^n$ with minimum Hamming distance $d$,
\begin{itemize}
\item if $\frac{d}{n} \geq 1 - \frac{1}{q}$, $|C| \leq 2qn$;
\item if $\frac{d}{n} > 1 - \frac{1}{q}$, $|C| \leq \frac{qd}{qd-(q-1)n}$.
\end{itemize}

We provide a similar bound in the Levenshtein metric based on Theorem~\ref{thm:bound}
by assuming the existence of the common supersequence of all codewords in the code.
Let $C \subseteq [q]^n$ be a code of minimum Levenshtein distance $d$,
and $v \in [q]^N$ be a word that contains every codeword in $C$ as a subsequence.
Since the word $v' = 12 \cdots q 1 2 \cdots q \cdots 1 2 \cdots q \in [q]^{qn}$
satisfies the condition of $v$, we can assume that $N \leq qn$.
Then, we use Theorem~\ref{thm:bound} in which we employ the above $v$ as a received word,
and set $t_\I = N-n$, $t_\D = 0$.
By the property of $v$, we have that
\begin{equation*}
  \ell \triangleq \max_{r \in \Sigma^{N}}|\lball(v, 0, N-n) \cap C| = |C|.
\end{equation*}
Thus, an upper bound on the code size is obtained by Lemma~\ref{lem:bound}.

\begin{theorem}\label{thm:plotkin}
  Let $C \subseteq \Sigma^n$ be a code of minimum Levenshtein distance $d$, where $|\Sigma| = q$.
  Suppose that there exists $v \in \Sigma^N$ that contains every codeword in $C$ as a subsequence.
  If $\frac{d}{2n} > 1  - \frac{n}{N}$, then
  \begin{equation*}
    |C| \leq \frac{Nd}{Nd - 2(N-n)n}.
  \end{equation*}
\end{theorem}
\begin{proof}
  Based on the above observation that $\ell = |C|$ for the word $v$,
  the statement follows simply by setting $t_\I = N-n$ and $t_\D = 0$ in Lemma~\ref{lem:bound}.
\end{proof}

When the word $v' = 12 \cdots q 1 2 \cdots q \cdots 1 2 \cdots q \in [q]^{qn}$ is employed in Theorem~\ref{thm:plotkin},
we have the statement that if $C \subseteq \Sigma^n$ 
satisfies $\frac{d}{2n} > 1 - \frac{1}{q}$,
then $|C| \leq \frac{qd}{qd - 2(q-1)n} =\frac{d/2n}{(d/2n) -(1 - 1/q)}$.
This result can be also obtained by the Plotkin bound in the Hamming metric
and the fact that if the minimum Levenshtein distance of $C$ is $d$, then the minimum Hamming distance of $C$ is at least $d/2$.
Thus, the theorem is effective if we could find $v$ of length less than $qn$.

Bukh and Ma~\cite{BM14} gave upper bounds on the size of codes. 
They first observed that if $\frac{d}{2n} > 1-\frac{1}{q}$, then $|C| \leq q$.
As the main result, they showed that for any non-negative integers $q, r, n$ with $q \geq 2$ and $n \geq q(10r)^{9r}$,
if $\frac{d}{2n} > 1 - \frac{1}{q} - c\cdot n^{-1/r}$, then $|C| \leq r + q+ 1$, where $c = \Theta(r^{-9}q^{1/r-2})$.
The first result gives a tighter bound than that obtained from Theorem~\ref{thm:plotkin} by using $v'$.
The main result of~\cite{BM14} is a tight bound when $n$ is sufficiently large. 
Compared with the results of~\cite{BM14},
Theorem~\ref{thm:plotkin} provides a bound for $n < q(10r)^{9r}$ or $N < qn/(1+c \cdot qn^{-1/r})$,
although the existence of a supersequence of length $N$ is assumed.





%

\appendix

\section{On the arguments of~\cite{WZ18}}\label{sec:wz17}

First, recall the arguments in the proof of a Johnson-type bound of~\cite{WZ18}.
Let $C \subseteq \Sigma^n$ be a code of minimum Levenshtein distance $d$.
For a given word $r \in \Sigma^N$, let $\lball(r, t)$ be the set of words that are within Levenshtein distance $t$ from $r$.
As in the proof of Lemma~\ref{lem:bound},
consider the set $\mathcal{L}' \triangleq \lball(r, t) \cap C = \{ c_1, c_2, \dots, c_\ell\}$.
For each $i \in [\ell]$, define $\mathcal{D'}^{(i)}$ and $\mathcal{E'}^{(i)}$ to be sets of positions of deletions
from $c_i$ and $r$, respectively, of the smallest size to obtain the same word.
By definition, $\mathcal{D'}^{(i)} \subseteq [n]$ and $\mathcal{E'}^{(i)} \subseteq [N]$.
Note that the Levenshtein distance between $c_i$ and $r$
is calculated as $\ldis(c_i, r) = |\mathcal{D'}^{(i)}| + |\mathcal{E'}^{(i)}|$.
In the proof, the relation that
\begin{align*}
  & \sum_{i \in [\ell]} \! \sum_{j \in [\ell]\setminus\{i \}}\!\! \ldis(c_i,c_j) \\
  & \!   \leq \!  \sum_{i \in [\ell]} \! \sum_{j \in [\ell]\setminus\{i \}} \!
  \left| \mathcal{D'}^{(i)} \! \cup \! \mathcal{D'}^{(j)} \! \setminus \!
  \{ k \! : \! k \in \mathcal{D'}^{(i)} \! \cap \! \mathcal{D'}^{(j)} \! \wedge  \!a^{i}_k \! = \! a^{j}_k\} \right| 
  + \left| (\mathcal{E'}^{(i)} \cup \mathcal{E'}^{(j)}) \setminus (\mathcal{E'}^{(i)} \cap \mathcal{E'}^{(j)}) \right|
\end{align*}
is used in the analysis.
The above inequality seems to be given by assuming the relation that
\begin{align}
  \ldis(c_i,c_j) \leq \left| \mathcal{D'}^{(i)} \! \cup \! \mathcal{D'}^{(j)} \! \setminus \!
\{ k \! : \! k \in \mathcal{D'}^{(i)} \! \cap \! \mathcal{D'}^{(j)} \! \wedge  \!a^{i}_k \! = \! a^{j}_k\} \right|
+ \left| (\mathcal{E'}^{(i)} \cup \mathcal{E'}^{(j)}) \setminus (\mathcal{E'}^{(i)} \cap \mathcal{E'}^{(j)}) \right| \triangleq u(i,j).\label{eq:wzrelation}
\end{align}
However, the relation $\ldis(c_i,c_j) \leq u(i,j)$  does not hold in general.
We present several counterexamples to (\ref{eq:wzrelation}).

Consider the codewords $c_1 = 000000, c_2 = 011100, c_3=100011$ and the received word $r_1 = 01100$.
We can choose such that
$\mathcal{D'}^{(1)}=\mathcal{D'}^{(3)} = \{4,5,6\}, \, \mathcal{D'}^{(2)}=\{4\}, \,
\mathcal{E'}^{(1)}=\{2,3\}, \, \mathcal{E'}^{(2)}=\emptyset, \, \mathcal{E'}^{(3)}=\{1,2\}$.
Fig.~\ref{fig:example} displays the codewords in $C$ together with $r$
in which $\mathcal{D'}^{(i)}$ and $\mathcal{E'}^{(j)}$ are shown by boxes.

\begin{figure}[h]
\begin{align*}
  c_1 & = 0 \ \ 0 \ \ 0 \ \boxed{0 \ \ 0 \ \ 0} & r_1 = \ 0 \ \boxed{1 \ \ 1} \ 0 \ \ 0\\
  c_2 & = 0 \ \ 1 \ \ 1 \ \boxed{ 1 } \ 0 \ \ 0 & r_1 = \ 0 \ \ 1 \ \ 1 \ \ 0 \ \ 0\\
  c_3 & = 1 \ \ 0 \ \ 0 \ \boxed{0 \ \ 1 \ \ 1} & r_1 = \boxed{0 \ \ 1} \ 1 \ \ 0 \ \ 0
\end{align*}
\caption{Codewords $c_1, c_2, c_3$ and the received word $r_1$, where the boxes represent the sets $\mathcal{D'}^{(i)}$ and $\mathcal{E'}^{(j)}$.}
\label{fig:example}
\end{figure}
We can see that $u(1,2) = 5, u(1,3) = 4, u(2,3)= 5$,
whereas $\ldis(c_1,c_2) = \ldis(c_1,c_3) = \ldis(c_2,c_3) = 6$.

Another simpler example is that $c_4 = 000111222, c_5 = 222111000, r_2 = 000$.
Then, $\mathcal{D'}^{(4)} = \{4,5,6,7,8,9\}, \mathcal{D'}^{(5)}=\{1,2,3,4,5,6\}$, and $\mathcal{E'}^{(4)}=\mathcal{E'}^{(5)}=\emptyset$.
It holds that $\ldis(c_4,c_5) =12$ and $u(4,5) =6$.
\begin{figure}[h]
\begin{align*}
  c_4 & = \ 0 \ \ 0 \ \ 0  \ \boxed{ 1 \ \ 1 \ \ 1 \ \ 2 \ \ \ 2 \ \ 2 } & r_2 = \ 0  \ \ 0 \ \ 0\\
  c_5 & = \boxed{ 2 \ \ 2 \ \ 2 \ \ 1 \ \ 1 \ \ 1 } \  0 \ \ 0 \ \ 0  & r_2 = \  0  \ \ 0 \ \ 0\\
\end{align*}
\caption{Codewords $c_4, c_5$ and the received word $r_2$.}
\label{fig:exampleq}
\end{figure}

As seen in the above, (\ref{eq:wzrelation}) does not hold in general. 

\section*{Acknowledgment}
This work was supported in part by JSPS Grants-in-Aid for Scientific Research Numbers 16H01705, 17H01695, and 18K11159.





\bibliographystyle{abbrv}
\bibliography{mybib}

\end{document}